\documentclass[reqno]{amsart}%
\usepackage{amsfonts}
\usepackage{verbatim}
\usepackage{xcolor}
\usepackage{amsmath}
\usepackage{amssymb}
\usepackage{graphicx}
\usepackage{accents}%
\providecommand{\U}[1]{\protect\rule{.1in}{.1in}}
\newtheorem{theorem}{Theorem}[section]
\newtheorem{condition}[theorem]{Hypothesis}

\newtheorem{proposition}[theorem]{Proposition}
\newtheorem{remark}[theorem]{Remark}

\numberwithin{equation}{section}
\begin{document}
\title[KdV equation]{KdV equation beyond standard assumptions on initial data\thanks{To appear in
Physica D: Nonlinear Phenomena}}
\author{Alexei Rybkin}
\address{Department of Mathematics and Statistics, University of Alaska Fairbanks, PO
Box 756660, Fairbanks, AK 99775}
\email{arybkin@alaska.edu}
\thanks{AR is supported in part by the NSF grant DMS-1411560}
\date{October, 2017}

\begin{abstract}
We show that the Cauchy problem for the KdV equation can be solved by the
inverse scattering transform (IST)\ for any initial data bounded from below,
decaying sufficiently rapidly at $+\infty$, but unrestricted otherwise. Thus
our approach doesn't require any boundary condition at $-\infty$.

\end{abstract}
\dedicatory{We dedicate this paper to the memory of Ludwig Faddeev, one of the founders of
soliton theory.}\maketitle
\tableofcontents

\section{Introduction}

This note is motivated in part by the open question posed by Vladimir Zakharov
in his plenary talk in July of 2016 at the XXXV Workshop on Geometric Methods
in Physics in Bialowieza, Poland. Zakharov stated the problem of understanding
the \emph{KdV equation}%

\begin{equation}%
\begin{cases}
\partial_{t}u-6u\partial_{x}u+\partial_{x}^{3}u=0\\
u(x,0)=q(x)
\end{cases}
\label{KdV}%
\end{equation}
with generic bounded but not decaying initial data $q$. He specifically
pointed out that (\ref{KdV}) no longer has (finite) conservation laws while
existence of infinitely many such laws is one of the main features of any
completely integrable system. A similar question was stated in his 2016 paper
\cite{Zakharovetal16}:\ 'Suppose that the initial condition of (\ref{KdV}) is
a bounded function, which is neither rapidly vanishing nor periodic. What is
its behavior under time evolution?'

To show how nontrivial this problem is let us put it in the historic context.
For smooth rapidly decaying $q$'s (\ref{KdV}) was solved in closed form in the
short 1967 paper \cite{GGKM67} by Gardner-Greene-Kruskal-Miura (GGKM) . As is
well-known, the paper \cite{GGKM67} introduces what we now call the\emph{
inverse scattering transform }(IST), one of the major discoveries in the
twentieth century mathematics. This paper was immediately followed by
\cite{Lax68} where the famous\emph{ Lax pair} first appeared. More
specifically, associate with $u\left(  x,t\right)  $ in (\ref{KdV}) two linear
operators, called the Lax pair,%
\begin{equation}
L\left(  t\right)  =-\partial_{x}^{2}+u\left(  x,t\right)  \text{ (the
\emph{Schr\"{o}dinger operator}),} \label{schrodinger}%
\end{equation}%
\[
P\left(  t\right)  =-4\partial_{x}^{3}+6u\left(  x,t\right)  \partial
_{x}+3\left(  \partial_{x}u\left(  x,t\right)  \right)  .
\]
The main observation made in \cite{Lax68} is that the KdV equation can be
represented as%
\begin{equation}
\partial_{t}L\left(  t\right)  =\left[  P\left(  t\right)  ,L\left(  t\right)
\right]  ,\text{ (the \emph{Lax representation})} \label{lax system}%
\end{equation}
which immediately implies that if $u$ solves (\ref{KdV}) then $L\left(
t\right)  $ is unitary equivalent to $L\left(  0\right)  =:\mathbb{L}%
_{q}=-\partial_{x}^{2}+q\left(  x\right)  .$ The latter means that the
spectrum of $L\left(  t\right)  $ is preserved under the KdV flow. This in
turn implies that $\psi$ and $\partial_{t}\psi-P\left(  t\right)  \psi$ are
both eigenfunctions (of discrete or continuos spectrum) associate with the
same point of spectrum. While the Lax representation (\ref{lax system}) is a
manifestation of a very specific structure of the KdV equation,
(\ref{lax system}) alone is not of much help to solving the Cauchy problem
(\ref{KdV}) as finding $\psi$ is not any easier than solving (\ref{KdV}). The
main reason why the Lax representation works is that in some important cases
direct computation of $\psi$ can be effectively circumvented. It is the case
in the original GGKM setting. More specifically, if we assume that (\ref{KdV})
has a smooth rapidly decaying solution $u\left(  x,t\right)  $ for all
$t\geq0$ then the Jost solutions $\psi_{\pm}$ remain Jost under the KdV flow.
This readily implies that the scattering data for the pair $\left(  L\left(
t\right)  ,\mathbb{L}_{0}\right)  $ evolves in time by a very simple law. The
solution $u\left(  x,t\right)  $ to (\ref{KdV}) for each $t>0$ is now obtained
by the formula%
\begin{equation}
u\left(  x,t\right)  =-2\partial_{x}^{2}\log\tau\left(  x,t\right)  ,
\label{tau}%
\end{equation}
where $\tau$ is the so-called \emph{Hirota tau-function} introduced in
\cite{Hirota71} and it admits an explicit representation in terms of the
scattering data\footnote{Will be given later.}. The solution has a relatively
simple and by now well understood wave structure of running (finitely many)
solitons accompanied by radiation of decaying waves (see e.g.
\cite{EckhuasSchuur83}, \cite{Schuur86}).

Another equally important case is when $q$ is periodic. Like in the previous
case the problem (\ref{KdV}) is a priori well-posed and a unique periodic
solution to (\ref{KdV}) exists. While it was quite clear from the beginning
that the GGKM approach should work but it was not until 1974 when the actual
IST was found in the periodic context by a considerable effort of such top
experts as Dubrovin, Flascka, Its, Marchenko, Matveev, McKean, Novikov,
Trubowitz, to name just few. We refer to the historically very first survey
\cite{DubMatNov76} by Dubrovin-Matveev-Novikov and the 2003 Gesztesy-Holden
book \cite{GesHold03} where a complete history is given. The periodic IST\ is
quite different from the GGKM\ one and is actually the \emph{inverse spectral
transform} (also abbreviated as IST) since it relies on the Floquet theory for
$\mathbb{L}_{q}$ and analysis of Riemann surfaces and hence is much more
complex than the rapidly decaying case. The time evolution of the spectral
data is nevertheless simple (but not simple to derive) and the solution
$u\left(  x,t\right)  $ is given essentially by the same formula (\ref{tau}),
frequently referred to as the \emph{Its-Matveev formula} \cite{ItsMat75}, but
$\tau$ is a multidimensional\footnote{Infinite dimensional in general.}
theta-function of real hyperelliptic algebraic curves explicitly computed in
terms of spectral data of the associated Dirichlet problem for $\mathbb{L}%
_{q}$. It is therefore very different from the rapidly decaying case. The main
feature of a periodic solution is its quasi-periodicity in time $t$.

We have outlined two main classes of initial data $q$ in (\ref{KdV}) for which
a suitable form of the IST\ was found during the initial boom followed by
\cite{GGKM67}. We emphasize again that existence of the Lax pair merely means
only that the KdV flow is isospectral but it doesn't in general offer an
algorithm to find the solution. It is the simple law of time evolution of the
scattering/spectral data that makes the IST work in these two cases. That is
why Krichever and Novikov claim in \cite{KrichNovikov99} that (\ref{KdV}) is
completely integrable essentially only in these two cases. In fact, the
question if (\ref{KdV})\footnote{Or any other integrable system.} could be
solved by a suitable IST outside of these two classes, has been raised in one
form or another by (in chronological order) McLeod-Olver \cite{McLO83},
Ablowitz-Clarkson \cite{AC91}, Marchenko \cite{MarchenkoWhatIteg91},
Krichever-Novikov \cite{KrichNovikov99}, Deift \cite{DeiftOpenProb08}, Matveev
\cite{MatveevOpenProblems}, and Zakharov \cite{Zakharovetal16} to name just a
few. These authors also expand on many challenges and complications that arise
and some regard it as a major unsolved problem.

We give a complete answer to the following question: Assuming rapid decay of
$q\left(  x\right)  $ at $+\infty$, what conditions do we have to impose at
$-\infty$ for (\ref{KdV}) to be well-posed in a certain sense and solvable by
a suitable IST? We show that the only condition to be imposed is that $q$ is
bounded from below. More specifically, we assume the following condition.

\begin{condition}
\label{Cond}$q\left(  x\right)  $ is a real, locally bounded function such that

\begin{enumerate}
\item For some $0\leq h<\infty$, $q\left(  x\right)  \geq-h^{2}$ (boundedness
from below);

\item For some $\alpha>4$, $q\left(  x\right)  =O\left(  x^{-\alpha}\right)
,x\rightarrow+\infty$ (decay at $+\infty$).
\end{enumerate}
\end{condition}

We call such $q$ \emph{step-type}. Thus any $q$ subject to Hypothesis
\ref{Cond} is bounded from below, decays sufficiently fast at $+\infty$ but is
arbitrary otherwise resulting in a much more complicated spectrum. The general
spectral theory of second-order ordinary differential operators says that the
negative spectrum of $\mathbb{L}_{q}$ has multiplicity one but could be of any
type (including absolutely continuous (a.c.)) and the positive spectrum has
a.c. component filling $\left(  0,\infty\right)  $ but need not be uniform
(however no embedded bound states\footnote{This is due to fast decay at
$+\infty$ which rules out solutions square integrable at $+\infty$.}).

Note that under our conditions neither well-posedness nor IST are a priori
available and we have to deal with both. Our approach is based upon the theory
of Hankel operators (see subsection \ref{H}). Our Hankel operator is unitary
equivalent to the well-known Marchenko integral operator but particularly
convenient for limiting procedures we crucially rely on. Following our
\cite{GruRyb-prepr-13}, we first introduce the one-sided scattering theory
from the right. The scattering data can be conveniently represented in terms
of the reflection coefficient $R$ from the right and certain positive measure
$\rho$ (see subsection \ref{scattering theory}) via%
\begin{equation}
\varphi_{x,t}(k)=\xi_{x,t}(k)R(k)+\int_{0}^{h}\frac{\xi_{x,t}(is)\,d\rho
(s)}{s+ik}, \label{eq9.9}%
\end{equation}
where%
\begin{equation}
\xi_{x,t}(k):=\exp\{i(8k^{3}t+2kx)\}\ \ \ \text{(cubic exponential).}
\label{cubic exp}%
\end{equation}
This function $\varphi_{x,t}$ appears as the symbol of our Hankel operator
$\mathbb{H}(\varphi_{x,t})$ which solely carries over the scattering data and
the variables $\left(  x,t\right)  $ in the KdV equation. In particular, if
$q$ is rapidly decaying also at $-\infty$ then $\rho$ becomes discrete and the
integral in (\ref{eq9.9}) becomes a finite sum.

Our main result is the following theorem.

\begin{theorem}
[Main Theorem]\label{MainThm} Suppose that the initial data $q$ in (\ref{KdV})
satisfy Hypothesis \ref{Cond}. Let%
\[
q_{b}\left(  x\right)  =\left\{
\begin{array}
[c]{cc}%
0, & x<b\\
q\left(  x\right)  , & x\geq b
\end{array}
\right.
\]
and denote by $u_{b}(x,t)$ the (necessarily unique) classical\footnote{I.e.,
at least three times continuously differentiable in $x$ and once in $t$.}
solution of (\ref{KdV}) with data $q_{b}$. Then for every $x$ and $t>0$ the
solutions $u_{b}(x,t)$ converge\footnote{In fact, uniform convergence on
compact sets of $\left(  x,t\right)  $ takes place but do not need it here. We
hope to present much more specific statements about the convergence
elsewhere.} to some $u(x,t)$ as $b\rightarrow-\infty$ which is also a
classical solution to the KdV equation. Moreover,%
\begin{equation}
u(x,t)=-2\partial_{x}^{2}\log\det\left(  1+\mathbb{H}(\varphi_{x,t})\right)  ,
\label{det_form}%
\end{equation}
with $\varphi_{x,t}$ defined by (\ref{eq9.9}), where the infinite determinant
is understood in the classical Fredholm sense.
\end{theorem}

Note that (\ref{KdV}) with data $q_{b}=\left.  q\right\vert _{\left(
b,\infty\right)  }$ is well-posed \cite{Bourgain93}, \cite{CohenKappSIAM87}
and therefore Theorem \ref{MainThm} also says that (\ref{KdV}) with data $q$
subject to Hypothesis \ref{Cond} is globally well-posed in the following
sense:\ classical solutions $q_{n}\left(  x,t\right)  $ with compactly
supported initial data $q_{n}\left(  x\right)  $ converge to a classical
solution $u\left(  x,t\right)  $ uniformly on any compact $x$-domain for any
$t>0$ and independently of the choice of $q_{n}\left(  x\right)  $
approximating $q\left(  x\right)  $. This definition is consistent with
\cite{KapTop06}, where it is also emphasized that existence implies uniqueness
and certain continuous dependence on the data. For general background reading
on well-posedness we refer the interested reader to the book \cite{Tao06} and
literature cited therein. For results on well-posedness of the KdV equation in
Sobolev spaces obtained by IST see the book \cite{KapBook03} and in weighted
Sobolev spaces see the recent \cite{Masperoetal16}. We are unaware of any
well-posedness results on the KdV equation with unrestricted behavior at
$-\infty$.

The main reason why our Hankel operator approach works is that it allows us to
show that classical solutions for restricted data $q_{b}$ given by
(\ref{det_form}) suitably converge as $b\rightarrow-\infty$ to the classical
solution of the KdV given by the same formula (\ref{det_form}).

Our result includes as particular cases, all $q$'s approaching a constant at
$-\infty$ (considered first in physical literature and rigorously in 1976 by
Hruslov\footnote{Also transcripted as Khruslov.} \cite{Hruslov76}) and a
periodic function (considered in 1994 by Kotlyarov-Hruslov \cite{KK94}). But
it was not until very recently when a compete rigorous investigation of
(\ref{KdV}) with such initial profiles and their generalizations was done by
Teschl with his collaborators (see e.g. \cite{TeschlRarefaction16},
\cite{Egorova11}, \cite{Egorovaetal13}, \cite{TeschlShock2016}, \cite{MLT12}).
We discuss some of their results in Section \ref{conclusions} where we also
give a brief review of some other results on nonclassical situations.

The paper is organized as follows. In Section \ref{classical IST} we discuss
the classical IST and give the solution to (\ref{KdV}) in terms of the Hankel
operator. In Section \ref{step type} we discuss the scattering theory for
potentials satisfying to Hypothesis \ref{Cond}. In Section \ref{main section}
we prove Theorem \ref{MainThm}. The last section is devoted to some relevant
discussions and connections of our results to those of others.

\section{Classical IST\ and Hankel operators}

\subsection{The classical IST\label{classical IST} (\cite{AC91,NPZ})}

For the \emph{Cauchy problem for the KdV equation }(\ref{KdV}) with real
rapidly decaying $q$'s the IST method consists, as the standard Fourier
transform method, of\ the three steps:

\textbf{Step 1. }\emph{Direct transform}: $q\left(  x\right)  \longrightarrow
S_{q},$ where $S_{q}$ is a new set of variables turning (\ref{KdV}) into a
simple order 1 linear ODE for $S_{q}(t)$ with the initial condition
$S_{q}(0)=S_{q}$.

\textbf{Step 2. }\emph{Time evolution}: $S_{q}\longrightarrow S_{q}\left(
t\right)  .$

\textbf{Step 3. }\emph{Inverse transform}: $S_{q}\left(  t\right)
\longrightarrow q(x,t).$

The set $S_{q}$ is formed as follows. Associate with $q$ the full line
\emph{Schr\"{o}dinger operator} $\mathbb{L}_{q}=-\partial_{x}^{2}+q(x)$. As
well-known, $\mathbb{L}_{q}$ is self-adjoint on $L^{2}:=$ $L^{2}\left(
\mathbb{R}\right)  $ and its spectrum consists of finitely many simple
(negative) \emph{bound states} $\{-\kappa_{n}^{2}\}$, and a twofold
\emph{absolutely continuous (a.c.) spectrum} filling $\mathbb{R}_{+}:=\left(
0,\infty\right)  $. The \emph{Schr\"{o}dinger equation} $\mathbb{L}_{q}%
\psi=k^{2}\psi$ has two \emph{(Jost) solutions}: $\psi_{\pm}(x,k)=e^{\pm
ikx}+o(1),\;x\rightarrow\pm\infty.$ The Jost solutions are analytic for
$\operatorname{Im}k>0$ and continuous for $\operatorname{Im}k\geq0$. Moreover,%
\begin{equation}
\psi_{\pm}(x,k)=e^{\pm ikx}\left(  1\pm\frac{i}{2k}\int_{x}^{\pm\infty
}q+O\left(  \frac{1}{k^{2}}\right)  \right)  ,\ \ \ k\rightarrow
\infty,\;\operatorname{Im}k\geq0, \label{eq6.4}%
\end{equation}
and
\begin{equation}
\psi_{\pm}(x,-k)=\overline{\psi_{\pm}(x,k)},\;k\in\mathbb{R}. \label{eq6.5}%
\end{equation}
The pair $\{\psi_{+},\overline{\psi_{+}}\}$ forms a fundamental set and
hence\footnote{We call (\ref{basic scatt identity}) the (right)\emph{\ basic
scattering relation.} Similarly, we define the left one by using $\left\{
\psi_{-},\overline{\psi_{-}}\right\}  $.}%
\begin{align}
T(k)\psi_{-}(x,k)  &  =\overline{\psi_{+}(x,k)}+R(k)\psi_{+}(x,k),\;k\in
\mathbb{R},\label{basic scatt identity}\\
&  \text{(\emph{basic scattering identity})}\nonumber
\end{align}
with some $T$ and$\ R$ called the \emph{transmission }and
(right)\emph{\ reflection coefficients }respectively. $R\left(  k\right)  $
has important properties \cite{DeiftTrub}:%
\begin{equation}
R\left(  -k\right)  =\overline{R\left(  k\right)  }\text{ (symmetry)}
\label{props of R}%
\end{equation}

\begin{equation}
\left\vert R\left(  k\right)  \right\vert <1,\ k\neq0,\text{ (contraction)}
\label{props of R2}%
\end{equation}%
\begin{equation}
R\left(  k\right)  =o\left(  1/k\right)  ,\left\vert k\right\vert
\rightarrow\infty\text{ (decay)} \label{props of R3}%
\end{equation}%
\begin{equation}
R\in C\left(  \mathbb{R}\right)  \text{ (continuity).} \label{props of R4}%
\end{equation}
\ 

Associate with $q$ the \emph{scattering data}%
\begin{equation}
S_{q}:=\left\{  R\left(  k\right)  ,\ k\geq0,\;\left(  \kappa_{n}%
,c_{n}\right)  ,1\leq n\leq N\right\}  , \label{classical scat data}%
\end{equation}
where $c_{n}$'s are positive numbers called \emph{norming constants} of bound
states $-\kappa_{n}^{2}$. In terms of Jost solutions $\psi_{\pm}$ one has%
\begin{equation}
R=\frac{W(\overline{\psi_{+}},\psi_{-})}{W(\psi_{-},\psi_{+})}\ \ \ (W\left(
f,g\right)  :=fg^{\prime}-f^{\prime}g),\ \ c_{n}=\left(  \int\left\vert
\psi_{+}(x,i\kappa_{n})\right\vert ^{2}dx\right)  ^{-1}. \label{eq6.8}%
\end{equation}
and Step 1 is solved. As well-known, $S_{q}$ determines $q$ uniquely. It is
the fundamental classical fact that under the KdV flow the scattering data
evolves in time as follows%
\begin{equation}
S_{q}(t)=\left\{  R(k)\exp8ik^{3}t,\ k\geq0,\;\left(  \kappa_{n},c_{n}%
\exp8\kappa_{n}^{3}t\right)  ,1\leq n\leq N\right\}  \label{time evol}%
\end{equation}
which solves Step 2. We emphasize that the Lax pair considerations do not
imply an explicit time evolution $\psi_{\pm}(x,t,k)$ for Jost solutions but
does imply that so for quantities (\ref{eq6.8}).

Step 3 amounts to solving the\emph{\ inverse scattering problem} of recovering
the potential $u\left(  x,t\right)  $ from $S_{q}(t)$ and can be done in many
ways. Historically, the first one is due to Gelfand-Levitan-Marchenko and it
requires solving an integral (Marchenko) equation. The most powerful one is
based on the \emph{Riemann-Hilbert problem} which is solved by means of
singular integral equations (cf. Deift-Zhou \cite{DZ93} or recent
Grunert-Teschl \cite{GT09} for a streamlined exposition of \cite{DZ93}). Our
approach also starts out from a Riemann-Hilbert problem (the basic scattering
relation (\ref{basic scatt identity})) which we solve in terms of Hankel operators.

\subsection{Hankel operators and the IST\label{H}}

A Hankel operator is an infinite-dimensional analog of a Hankel matrix, a
matrix whose $(j,k)$ entry depends only on $j+k$. I.e. a matrix $\Gamma$ of
the form
\[
\Gamma=\left(
\begin{array}
[c]{cccc}%
\gamma_{1} & \gamma_{2} & \gamma_{3} & ...\\
\gamma_{2} & \gamma_{3} & ... & \\
\gamma_{3} & ... &  & \\
... &  &  & \gamma_{n}%
\end{array}
\right)  .
\]
The immediate Hilbert space generalization of a Hankel matrix is an integral
operator on $L^{2}(\mathbb{R}_{+})$ whose kernel depends on the sum of the
arguments%
\begin{equation}
(\mathbb{H}f)(x)=\int_{0}^{\infty}h(x+y)f(y)dy,\;f\in L^{2}(0,\infty
),\;x\geq0, \label{eq4.10}%
\end{equation}
and it is in this form that Hankel operators typically appear in the inverse
scattering formalism and are referred to as Marchenko's operator. The form
(\ref{eq4.10}) however does not prove to be convenient for our purposes and in
fact it is not used much in the Hankel operator community either. Instead, we
consider Hankel operators on Hardy spaces. (see e.g. excellent books
\cite{Nik2002}, \cite{Peller2003} for more information and numerous
references). We recall that a function $f$ analytic in $\mathbb{C}^{\pm}$ is
in the Hardy space $H_{\pm}^{2}$ if
\[
\sup_{y>0}\int_{-\infty}^{\infty}\left\vert f(x\pm iy)\right\vert
^{2}\ dx<\infty.
\]
It is particularly important that $H_{\pm}^{2}$ is a Hilbert space with the
inner product induced from $L^{2}$:
\[
\langle f,g\rangle_{H_{\pm}^{2}}=\langle f,g\rangle_{L^{2}}=\left\langle
f,g\right\rangle =\int_{-\infty}^{\infty}f\left(  x\right)  \bar{g}\left(
x\right)  dx.
\]
It is well-known that $L^{2}=H_{+}^{2}\oplus H_{-}^{2},$ the orthogonal
(Riesz) projection $\mathbb{P}_{\pm}$ onto $H_{\pm}^{2}$ being given by%
\begin{equation}
(\mathbb{P}_{\pm}f)(x)=\pm\frac{1}{2\pi i}\int_{-\infty}^{\infty}%
\frac{f(s)\ ds}{s-(x\pm i0)}. \label{proj}%
\end{equation}
Notice that for any $f\in H_{+}^{2}$ and $\lambda\in\mathbb{C}^{+}$%
\begin{equation}
\mathbb{P}_{-}\frac{f\left(  \cdot\right)  }{\cdot-\lambda}=\mathbb{P}%
_{-}\frac{f(\cdot)-f(\lambda)}{\cdot-\lambda}+\mathbb{P}_{-}\frac{f(\lambda
)}{\cdot-\lambda}=\frac{f(\lambda)}{\cdot-\lambda}. \label{P_}%
\end{equation}

We will also need $H_{\pm}^{\infty}$, the algebra of analytic functions
uniformly bounded in $\mathbb{C}^{\pm}$.

Let $(\mathbb{J}f)(x)\overset{\operatorname*{def}}{=}f(-x)$ be the operator of
reflection. It is clearly an isometry on $L^{2}$ intertwining the Riesz
projections%
\begin{equation}
\mathbb{JP}_{\mp}=\mathbb{P}_{\pm}\mathbb{J}. \label{eq4.9}%
\end{equation}

Given $\varphi\in L^{\infty}$ the operator $\mathbb{H}(\varphi):H_{+}%
^{2}\rightarrow H_{+}^{2}$ is called \emph{Hankel }if%
\begin{equation}
\mathbb{H}(\varphi)f\overset{\operatorname*{def}}{=}\mathbb{JP}_{-}\varphi
f,\ \ \ f\in H_{+}^{2}, \label{Hankel}%
\end{equation}
and $\varphi$ is called its \emph{symbol}.

It immediately follows from the definition (\ref{Hankel}) that%
\begin{equation}
\mathbb{H}(\varphi+h)=\mathbb{H}(\varphi)\text{ for any }h\in H_{+}^{\infty}.
\label{fi plus h}%
\end{equation}
meaning that only $\mathbb{P}_{-}\varphi$ (the so-called co-analytic) part of
the symbol matters.

Observe that If $\mathbb{J}\varphi=\overline{\varphi}$ then $\mathbb{H}%
(\varphi)$ is obviously selfadjoint. Note that $\mathbb{H}(\varphi)$ is
unitary equivalent to the operator $\mathbb{H}$ given by (\ref{eq4.10}),
$\varphi$ being the Fourier transform of $h$ in (\ref{eq4.10}).

For a given bounded operator $\mathbb{A}$ on a Hilbert space, we recall that
its $n$-th singular value $s_{n}\left(  \mathbb{A}\right)  $ is defined as the
$n$-th eigenvalue of the operator $\left(  \mathbb{A}^{\ast}\mathbb{A}\right)
^{1/2}$. We say that $\mathbb{A}$ is compact if $s_{n}\left(  \mathbb{A}%
\right)  \rightarrow0$ and we write $\mathbb{A\in}\mathfrak{S}_{\infty}$. If
$\Sigma_{n}s_{n}\left(  \mathbb{A}\right)  =:\left\Vert A\right\Vert
_{\mathfrak{S}_{1}}<\infty$ then $\mathbb{A}$ is called a \emph{trace class
}operator and we write $\mathbb{A\in}\mathfrak{S}_{1}.$

The membership of a Hankel operator $\mathbb{H}(\varphi)$ in the trace class
is determined by smoothness of its symbol $\varphi$. In particular, the
following criterion holds (see e.g. Theorem 9.1 and 9.2 in
\cite{GruRyb-prepr-13}).

\begin{proposition}
[Adamyan-Arov-Krein]\label{trace AAK}If a bounded function $\varphi$ is twice
differentiable on $\mathbb{R}$ then $\in\mathfrak{S}_{1}$.
Moreover\footnote{Recall $\left\Vert f\right\Vert _{\infty}=\sup\left\vert
f\left(  x\right)  \right\vert ,x\in\mathbb{R}.$},%
\[
\left\Vert \mathbb{H}(\varphi)\right\Vert _{\mathfrak{S}_{1}}\leq
const\left\Vert \varphi^{\prime\prime}\right\Vert _{\infty}.
\]

\end{proposition}

The proof of Proposition \ref{trace AAK} is based upon the seminal
Adamyan-Arov-Krein Theorem on singular numbers of Hankel operators (see, e.g.
\cite{Nik2002}, \cite{Peller2003} and the original literature cited therein).
A necessary and sufficient condition for $\mathbb{H}(\varphi)\in
\mathfrak{S}_{1}$ is found by Peller \cite{Peller2003}. We are not sure if the
beautiful Adamyan-Arov-Krein theory has been used in soliton theory.

\begin{proposition}
\label{lemma on Trace class} Let $\varphi_{1},\varphi_{2}$ are bounded and
suppose $\mathbb{H}(\varphi_{1}),\mathbb{H}(\varphi_{2})\in\mathfrak{S}_{1}$,
then $\mathbb{H}(\varphi_{1}\varphi_{2})\in\mathfrak{S}_{1}$.
\end{proposition}

This statement is of course well-known but for the reader's convenience we
give its elementary proof.

\begin{proof}
Note first that $\mathbb{H}(\varphi)$ is different from\footnote{In fact, in
the literature it is how Hankel operators is typically defined.}%
\[
\widetilde{\mathbb{H}}(\varphi)f=\mathbb{P}_{-}\varphi f,\ \ \ f\in H_{+}^{2}%
\]
only by isometry $\mathbb{J}$ and hence it is sufficient to prove the
statement for $\widetilde{\mathbb{H}}(\varphi)$. Since $\mathbb{P}%
_{+}+\mathbb{P}_{-}=I$, we then have%
\begin{align*}
\widetilde{\mathbb{H}}(\varphi_{1}\varphi_{2})f  &  =\mathbb{P}_{-}\varphi
_{1}\varphi_{2}f=\mathbb{P}_{-}\varphi_{1}\left(  \mathbb{P}_{+}%
+\mathbb{P}_{-}\right)  \varphi_{2}f\\
&  =\mathbb{P}_{-}\varphi_{1}\mathbb{P}_{+}\varphi_{2}f+\mathbb{P}_{-}%
\varphi_{1}\mathbb{P}_{-}\varphi_{2}f\\
&  =\widetilde{\mathbb{H}}(\varphi_{1})\mathbb{P}_{+}\varphi_{2}%
f+\mathbb{P}_{-}\varphi_{1}\widetilde{\mathbb{H}}(\varphi_{2})f.
\end{align*}
But $\left\Vert AKB\right\Vert _{\mathfrak{S}_{1}}\leq\left\Vert A\right\Vert
\left\Vert K\right\Vert _{\mathfrak{S}_{1}}\left\Vert B\right\Vert $ and
therefore%
\[
\left\Vert \widetilde{\mathbb{H}}(\varphi_{1}\varphi_{2})\right\Vert
_{\mathfrak{S}_{1}}\leq\left\Vert \widetilde{\mathbb{H}}(\varphi
_{1})\right\Vert _{\mathfrak{S}_{1}}\left\Vert \varphi_{2}\right\Vert
_{\infty}+\left\Vert \varphi_{1}\right\Vert _{\infty}\left\Vert
\widetilde{\mathbb{H}}(\varphi_{2})\right\Vert _{\mathfrak{S}_{1}}.
\]

\end{proof}

As mentioned before the Hankel operator appears in the classical IST as
Marchenko's integral operator. But the integral realization (\ref{eq4.10}) of
a Hankel operator has some serious technical disadvantages and for some
serious reasons is less popular in the Hankel operator community. On the other
hand, we have not seen Marchenko's operator written in the form (\ref{Hankel}).

We now demonstrate how convenient the definition (\ref{Hankel}) of the Hankel
operator is for solving (\ref{KdV}) for rapidly decaying initial data in
closed form.

Introduce%
\[
y_{\pm}\left(  x,k\right)  :=e^{\mp ikx}\psi_{\pm}(x,k)\text{ \ \ (Faddeev
functions)}%
\]
and rewrite the basic scattering identity (\ref{basic scatt identity}) in the
form%
\begin{equation}
Ty_{-}=\bar{y}_{+}+R\xi_{x}y_{+}, \label{eq6.15}%
\end{equation}
where $\xi_{x}(k):=e^{2ikx}$. Let us regard (\ref{eq6.15}) as a
Riemann-Hilbert problem of determining $y_{\pm}$ by given $T,R$ which we will
solve by Hankel operator techniques. For simplicity assume that there is only
one bound state $-\kappa_{0}^{2}$ with the norming constant $c_{0}$.

The function $Ty_{-}$ in (\ref{eq6.15}) is meromorphic in $\mathbb{C}^{+}$ as
a function of $k$ for each $x$ with a simple pole at $i\kappa_{0}$ and the
residue%
\begin{equation}
\operatorname*{Res}\limits_{k=i\kappa_{0}}T\left(  k\right)  y_{-}%
(x,k)=ic_{0}\xi_{x}\left(  i\kappa_{0}\right)  y_{+}\left(  x,i\kappa
_{0}\right)  . \nonumber\label{res}%
\end{equation}
Therefore \cite{DeiftTrub} for each fixed $x$%
\[
T\left(  k\right)  y_{-}(x,k)-1-\frac{ic_{0}\xi_{x}(i\kappa_{0})}%
{k-i\kappa_{0}}y_{+}(x,i\kappa_{0})\in H_{+}^{2}.
\]
Rearrange (\ref{eq6.15}) to read%

\begin{align}
&  T\left(  k\right)  y_{-}(x,k)-1-\frac{ic_{0}\xi_{x}(i\kappa_{0})}%
{k-i\kappa_{0}}y_{+}\left(  x,i\kappa_{0}\right) \nonumber\\
&  =\overline{\left(  y_{+}(x,k)-1\right)  }+\left(  R\xi_{x}\right)  \left(
k\right)  \left(  y_{+}(x,k)-1\right) \nonumber\\
&  +\left(  R\xi_{x}\right)  \left(  k\right)  -\frac{ic_{0}\xi_{x}%
(i\kappa_{0})}{k-i\kappa_{0}}y_{+}\left(  x,i\kappa_{0}\right)  .
\label{eq6.16}%
\end{align}
Noticing that the last term in (\ref{eq6.16}) is in $H_{-}^{2}$, we can apply
the Riesz projection $\mathbb{P}_{-}$ to (\ref{eq6.16}). Introducing
$Y:=y_{+}-1$, we have%
\begin{equation}
\mathbb{P}_{-}(\overline{Y}+R\xi_{x}Y)+\mathbb{P}_{-}R\xi_{x}-ic_{0}\xi
_{x}(i\kappa_{0})\ \frac{Y\left(  x,i\kappa_{0}\right)  }{\cdot-i\kappa_{0}%
}-\frac{ic_{0}\xi_{x}(i\kappa_{0})}{\cdot-i\kappa_{0}}=0. \label{eq6.17}%
\end{equation}
It follows from (\ref{eq6.4}) that $Y\in H_{+}^{2}$ for any $x\in\mathbb{R}$.
Due to the symmetry (\ref{eq6.5}), $\overline{Y}=\mathbb{J}Y$ and by
(\ref{eq4.9}) we have%
\begin{equation}
\mathbb{P}_{-}\overline{Y}=\mathbb{P}_{-}\mathbb{J}Y=\mathbb{JP}%
_{+}Y=\mathbb{J}Y. \label{eq6.18}%
\end{equation}
Note that by (\ref{P_})%
\begin{equation}
\frac{Y\left(  i\kappa_{0},x\right)  }{\cdot-i\kappa_{0}}=\mathbb{P}_{-}%
\frac{Y(\cdot,x)}{\cdot-i\kappa_{0}}. \label{eq6.19}%
\end{equation}
Inserting (\ref{eq6.18}) and (\ref{eq6.19}) into (\ref{eq6.17}), we obtain%
\[
\mathbb{J}Y+\mathbb{P}_{-}\left(  R\xi_{x}-\frac{ic_{0}\xi_{x}(i\kappa_{0}%
)}{\cdot-i\kappa_{0}}\right)  Y=-\mathbb{P}_{-}\left(  R\xi_{x}-\frac
{ic_{0}\xi_{x}(i\kappa_{0})}{\cdot-i\kappa_{0}}\right)  .
\]
Applying $\mathbb{J}$ to both sides of this equation yields%
\begin{equation}
(\mathbb{I}+\mathbb{H}(\varphi_{x}))Y=-\mathbb{H}(\varphi_{x})1,
\label{eq6.20}%
\end{equation}
where $\mathbb{H}(\varphi_{x})$ is the Hankel operator with symbol%
\[
\varphi_{x}(k)=R(k)\xi_{x}(k)+\frac{c_{0}\xi_{x}(i\kappa_{0})}{\kappa_{0}%
+ik}.
\]
Similarly, for $N$ bound states one has%

\[
\varphi_{x}(k)=R(k)\xi_{x}(k)+\sum_{n=1}^{N}\frac{c_{n}\xi_{x}(i\kappa_{n}%
)}{\kappa_{n}+ik}.
\]
By (\ref{props of R}), $\mathbb{J}\varphi_{x}=\overline{\varphi_{x}}$ and
hence $\mathbb{H}(\varphi_{x})$ is selfadjoint. It follows from
(\ref{props of R4}) that $\mathbb{H}(\varphi_{x})$ is compact for any $x$.
Note that $\mathbb{H}(\varphi_{x})1$ on the right hand side of (\ref{eq6.20})
should be interpreted as $\mathbb{H}(\varphi_{x})1=\mathbb{P}_{+}%
\overline{\varphi_{x}}\in H_{+}^{2}.$

It is now clear that if we show that (\ref{eq6.20}) is uniquely solvable and
$Y(x,k)$ is its solution then the potential $q\left(  x\right)  $ can be found
from (\ref{eq6.4}) by%
\begin{equation}
q(x)=\partial_{x}\lim2ikY(x,k),\ \ \ k\rightarrow\infty. \label{eq6.22}%
\end{equation}

Alternatively,%
\begin{equation}
q(x)=-2\partial_{x}^{2}\log\det\left(  1+\mathbb{H}(\varphi_{x})\right)  ,
\label{dyson}%
\end{equation}
where the determinant is understood in the classical Fredholm sense. In a
different (but equivalent) form (\ref{dyson}) has been known since Bargmann
and has been derived in a number of different ways (cf. e.g.
\cite{Dyson76,NPZ,Poppe84}). Loosely speaking, it follows from solving
(\ref{eq6.20}) by Cramer's rule.

Steps 2 and 3 of Subsection \ref{classical IST} now merely amounts to
replacing $\varphi_{x}$ with%
\[
\varphi_{x,t}\left(  k\right)  =R(k)\xi_{x,t}(k)+\sum_{n=1}^{N}\frac{c_{n}%
\xi_{x,t}(i\kappa_{n})}{\kappa_{n}+ik},
\]
where $\xi_{x,t}(k)=\exp i(8k^{3}t+2kx)$ solely carries the dependence on
$\left(  x,t\right)  $. Thus Steps 1-3 can now be put together in a compact
form%
\begin{equation}
q(x)\longrightarrow\mathbb{H}(\varphi_{x,t})\longrightarrow q(x,t),
\label{steps 1-3}%
\end{equation}
where $q(x,t)$ is explicitly given by%
\begin{equation}
q(x,t)=-2\partial_{x}^{2}\log\det\left(  1+\mathbb{H}(\varphi_{x,t})\right)  .
\label{classical sltn}%
\end{equation}
We will call $\mathbb{H}(\varphi_{x,t})$ the \emph{IST Hankel operator}.

For rapidly decaying initial data our Hankel operator approach is shorter but
completely equivalent to the classical treatment. However our edition
(\ref{steps 1-3})-(\ref{classical sltn}) of the IST turns to be a very
convenient starting point to extend IST far beyond standard assumptions on
initial data.

\section{Step-type potentials\label{step type}}

In this section we discuss scattering theory for step-type potentials
following our \cite{GruRyb-prepr-13}. But first we need to review some facts
from the classical Titchmarsh-Weyl theory.

\subsection{Titchmarsh-Weyl Theory and $m-$function (\cite{Titchmarsh62}%
)\label{m-function}}

The main point of this theory is that the problem
\[
\mathbb{L}_{q}u=\lambda u,\ x\in\mathbb{R}_{\pm},\ u\left(  \cdot
,\lambda\right)  \in L^{2}\left(  0,\pm\infty\right)  ,\lambda\in
\mathbb{C}^{+},
\]
has a unique (up to a multiplicative constant) solution $\Psi_{\pm}%
(x,\lambda)$, called a \emph{Weyl} \emph{solution} for broad classes of $q$'s
(called\emph{\ limit point case at }$\pm\infty$). Define then the
\emph{(Titchmarsh-Weyl) m-function} $m_{\pm}$ for $\left(  0,\pm\infty\right)
$ as follows:%
\begin{equation}
m_{\pm}\left(  \lambda\right)  =\pm\partial_{x}\log\Psi_{\pm}\left(
\pm0,\lambda\right)  ,\ \ \lambda\in\mathbb{C}^{+}. \label{m-funct}%
\end{equation}
One defines $m_{\pm}\left(  \lambda,a\right)  $ for $\left(  a,\pm
\infty\right)  $ in a similar way.

By the \emph{Borg-Marchenko uniqueness theorem} $\left\{  m_{+},m_{-}\right\}
$ recovers $q$ uniquely (see \cite{GestSimon2000} and the original literature
cited therein). While fundamentally important to spectral theory of OD
operators\footnote{The dependence of $m$ on $q$ is very intricate and
\ understood only in rather narrow classes.}, its role in scattering theory is
modest. Moreover, Steps 1-3 in the previous subsection with data
$S_{q}=\left\{  m_{+},m_{-}\right\}  $ don't work well
\cite{MarchenkoWhatIteg91}. For this reason the $m-$function is little known
in the soliton community.

In our approach the $m-$function plays a supporting but nevertheless crucial
role due to the following reasons: it's well-defined for any realistic $q$
(with no decay assumptions) including $q$'s subject to Hypothesis \ref{Cond},
and it's a \emph{Herglotz function}, i.e. it is analytic and maps
$\mathbb{C}^{+}$ analytically to $\mathbb{C}^{+}$. As is well-know, each such
function $f$ can be represented as%
\begin{equation}
f\left(  \lambda\right)  =a+b\lambda+\int_{-\infty}^{\infty}\left(  \frac
{1}{s-\lambda}-\frac{s}{1+s^{2}}\right)  d\mu\left(  s\right)
\label{Hergtotz}%
\end{equation}
with some%
\[
a\in\mathbb{C}^{+},\ \ b\geq0,\ \ d\mu\left(  s\right)  \geq0,\ \int_{-\infty
}^{\infty}\frac{d\mu\left(  s\right)  }{1+s^{2}}<\infty.
\]
If $m=m_{+}$ is given by (\ref{m-funct}) then $\mu$ is the spectral measure of
the Schr\"{o}dinger operator on $L^{2}\left(  0,\infty\right)  $ with a
Dirichlet boundary condition at $x=0.$ The latter implies that if the spectrum
is bounded from below than so is the support of $\mu$. It follows that $m$ can
be analytically extended into $\mathbb{C}^{-}$.

\begin{theorem}
\label{conver of m-functions}If $q,q_{n}$ are limit point case at $\pm
\infty\,$\ and $q_{n}\rightarrow q$ in $L_{\operatorname*{loc}}^{1}$ then
$m_{\pm}\left(  \lambda,q_{n}\right)  \rightarrow m_{\pm}\left(
\lambda,q\right)  $ uniformly on compacts away from the spectra.
\end{theorem}

A proof of this statement in the most general case is given in
\cite{GruRemRyb2015}.

Note that by definition (\ref{m-funct}) the $m-$function is a 1D
\emph{Dirichlet-to-Neumann map}.

\subsection{Scattering theory for step-type
potentials\label{scattering theory}}

The main feature of such potentials is that we can do one-sided\ scattering
theory replacing in (\ref{basic scatt identity}) the Jost solution $\psi_{-}$
with Weyl solution $\Psi_{-}$. This immediately yields%
\begin{equation}
R=W(\overline{\psi_{+}},\Psi_{-})/W(\Psi_{-},\psi_{+}) \label{gener R}%
\end{equation}
which is consistent with the classical reflection coefficient. While
properties (\ref{props of R})-(\ref{props of R4}) all hold for rapidly
decaying potentials, only (\ref{props of R}) holds for our step-type
potentials. The property (\ref{props of R2}) is replaced with $\left\vert
R\left(  k\right)  \right\vert \leq1$ but $\left\vert R\left(  k\right)
\right\vert =1$ may occur for almost all $k$. The properties
(\ref{props of R3})-(\ref{props of R4}) fail and this is a very serious
circumstance even for the powerful Riemann-Hilbert problem approach. In
\cite{GruRyb-prepr-13} we found what makes up for the lost properties:

\begin{theorem}
[Analytic split formula]\label{step like refl}Under Hypothesis \ref{Cond}
\begin{equation}
R\left(  k\right)  =R_{a}\left(  k\right)  +\xi_{a}^{-1}\left\{  A_{a}\left(
k\right)  -T_{a}\left(  k\right)  /y_{+}\left(  a,k\right)  \right\}  ,
\label{R-split}%
\end{equation}
where $R_{a},\ T_{a}$ are respectively the right reflection, transmission
coefficients from $q_{a}=\left.  q\right\vert _{\left(  a,\infty\right)  }$,
and
\begin{equation}
A_{a}\left(  k\right)  =2ik\ y_{+}\left(  a,k\right)  ^{-2}\ \left(
m_{-}(k^{2},a)+m_{+}(k^{2},a)\right)  ^{-1} \label{Aa}%
\end{equation}
is analytic in $\mathbb{C}^{+}$ except for $i\Delta:=\left\{  k:k^{2}%
\in\operatorname*{Spec}\mathbb{L}_{q}\cap\left(  -\infty,0\right)  \right\}
$. Moreover, (1)
\begin{equation}
R_{a}\left(  k\right)  =\frac{T_{a}\left(  k\right)  }{2ik}G_{a}\left(
k\right)  ,\ \ G_{a}\left(  k\right)  :=\int_{a}^{\infty}e^{-2iks}Q\left(
s\right)  ds, \label{R_0}%
\end{equation}
with some function $Q$ (independent of $a$) such that%
\begin{equation}
\left\vert Q\left(  s\right)  \right\vert \leq\left\vert q\left(  s\right)
\right\vert +const\int_{s}^{\infty}\left\vert q\right\vert ; \label{g}%
\end{equation}
(2) for $a$ large enough $y_{+}\left(  a,k\right)  ^{-1}\in H_{+}^{\infty}$;%
\[
T_{a}\left(  k\right)  =\dfrac{k+i\varkappa_{a}}{k-i\varkappa_{a}}g_{a}\left(
k\right)  ,\ \ \varkappa_{a}\geq0,
\]
where $g_{a}\in H_{+}^{\infty}$ has only one simple zero at $k=0$ in
$\mathbb{C}^{+}\cup\mathbb{R}$; (3) for $a$ large enough the jump
$A_{a}(is-0)-A_{a}(is+0)$ across $i\Delta$ is independent of $a$ and defines a
non-negative finite measure \
\begin{equation}
d\rho\left(  s\right)  :=i\left\{  A_{a}(is-0)-A_{a}(is+0)\right\}
ds/2\pi=\operatorname{Im}A_{a}(is+0)ds/\pi\label{drou}%
\end{equation}
supported on $\Delta\subseteq\lbrack0,h]$.
\end{theorem}

The set
\[
S_{q}=\{R,\rho\}
\]
plays the role of the classical scattering data (\ref{classical scat data}) in
our one sided scattering. The measure $\rho$ carries over the information on
the negative spectrum. In particular, if $q$ is rapidly decaying at both
$\pm\infty$ then $d\rho=%
{\displaystyle\sum}
c_{n}^{2}\delta\left(  s-\kappa_{n}\right)  ds$. Therefore we can call $\rho$
a smeared norming constant measure. If $q(x)$ is a pure step function, i.e.
$q(x)=-h^{2},\;x<0,\;q(x)=0,\;x\geq0$ then $\operatorname*{Spec}\left(
\mathbb{L}_{q}\right)  =(-h^{2},\infty)$ and purely a.c., $(-h^{2},0)$ and
$\left(  0,\infty\right)  $ being components of the spectrum with respective
multiplicities one and two. Moreover%
\[
R(k)=-\left(  \frac{h}{\sqrt{k^{2}}+\sqrt{k^{2}+h^{2}}}\right)  ^{2}%
,\ \ \rho\left(  s\right)  =\frac{1}{3\pi h^{2}}\left(  h^{3}-\left(
h^{2}-s^{2}\right)  ^{3/2}\right)  .
\]

The split (\ref{R-split}) is surprisingly effective. Its main feature is that
the analytic part $A_{a}$ of $R$ mimics the rough behavior of $R$ and carries
all the information about $\rho$ in the set of scattering data $S_{q}%
=\{R,\rho\}$. The rest is at least continuous and its smoothness is determined
by the decay of $q$ at $+\infty\,$. This is crucially used in developing the
IST for step-type initial data as all properties of $R$ required for the IST
are encoded in the analytic part $A_{a}$ of $A$ through the $m-$functions
$m_{\pm}$. Thus the $m-$function works behind the scene but in a crucial way.

\section{The IST\ for the KdV equation with step-type initial
data\label{main section}}

In this section we prove Theorem \ref{MainThm}. Since the KdV equation is
translation invariant by shifting $q$ (if needed) we may assume in Theorem
\ref{step like refl} that $a=0$. Moreover, to avoid unnecessary technicalities
we suppose that $\kappa_{0}=0$ (i.e. $q_{0}$ has no bound states). In this
case $T_{0}=g_{0}\in H_{+}^{\infty}$ and (\ref{R-split}) simplifies to read%
\begin{equation}
R\left(  k\right)  =R_{0}\left(  k\right)  +A\left(  k\right)  ,
\label{anal split}%
\end{equation}
with%
\begin{equation}
R_{0}\left(  k\right)  =f_{0}\left(  k\right)  G_{0}\left(  k\right)
,\ \ G_{0}\left(  k\right)  =\int_{0}^{\infty}e^{-2iks}Q\left(  s\right)  ds
\label{Rnot}%
\end{equation}
and%
\begin{align}
A\left(  k\right)   &  =2ik\ y_{+}\left(  0,k\right)  ^{-2}\ \left(
m_{-}(k^{2})+m_{+}(k^{2})\right)  ^{-1}-T_{0}\left(  k\right)  /y_{+}\left(
0,k\right) \nonumber\\
&  =\ f_{1}\left(  k\right)  \ \frac{2ik}{m_{-}(k^{2})+m_{+}(k^{2})}%
+f_{2}\left(  k\right)  , \label{A}%
\end{align}
where
\begin{equation}
f_{0}\left(  k\right)  =T_{0}\left(  k\right)  /2ik,\ \ f_{1}\left(  k\right)
=y_{+}\left(  0,k\right)  ^{-2},\ \ f_{2}\left(  k\right)  =-T_{0}\left(
k\right)  /y_{+}\left(  0,k\right)  \ \in H_{+}^{\infty}, \label{f's}%
\end{equation}
i.e. are analytic functions all bounded in the upper half plane. Throughout
this section we assume that $q$ in (\ref{KdV}) is subject to Hypothesis
\ref{Cond}.

\subsection{Fundamental properties of the IST Hankel operator}

\begin{theorem}
\label{th9.5} Let%
\[
\varphi_{x,t}(k)=\xi_{x,t}(k)R(k)+\int_{0}^{h}\frac{\xi_{x,t}(is)\,d\rho
(s)}{s+ik}.
\]
Under Hypothesis \ref{Cond} for any real $x$ and positive $t$ the operator
$\mathbb{H}(\varphi_{x,t})$

\begin{enumerate}
\item is selfadjoint (also holds for $t=0$),

\item has no eigenvalue equal $-1$,

\item its derivatives $\partial_{t}\mathbb{H}\left(  \varphi_{x,t}\right)
,\partial_{x}^{m}\mathbb{H}\left(  \varphi_{x,t}\right)  ,0\leq m\leq5,$ are
compact Hankel operators,

\item is trace class.
\end{enumerate}
\end{theorem}

\begin{proof}
Part (1) is trivial as clearly $\varphi_{x,t}(-k)=\overline{\varphi_{x,t}%
(-k)}$. Part (2) is the most difficult and it is proven in our
\cite{GruRyb-prepr-13}.

Let us show (3). It follows from (\ref{anal split}) that%
\[
\varphi_{x,t}(k)=\xi_{x,t}(k)R_{0}(k)+\left\{  \xi_{x,t}(k)A(k)+\int_{0}%
^{h}\frac{\xi_{x,t}(is)\,d\rho(s)}{s+ik}\right\}  .
\]
It follows from (\ref{fi plus h}) that%
\begin{equation}
\mathbb{H}\left(  \varphi_{x,t}\right)  =\mathbb{H}\left(  \widetilde{\varphi
}_{x,t}\right)  , \label{fi tilde}%
\end{equation}
with%
\[
\widetilde{\varphi}_{x,t}\left(  k\right)  =\xi_{x,t}R_{0}+\Phi_{x,t},
\]
where%
\begin{align}
\Phi_{x,t}\left(  k\right)   &  =\mathbb{P}_{-}\left(  \xi_{x,t}A\right)
\left(  k\right)  +\int_{0}^{h}\frac{\xi_{x,t}(is)\,d\rho(s)}{s+ik}%
\label{cont deform}\\
&  =\frac{i}{2\pi}\int_{\mathbb{R}}\ \frac{\xi_{x,t}\left(  z\right)  A\left(
z\right)  }{z-\left(  k-i0\right)  }\ dz+\int_{0}^{h}\frac{\xi_{x,t}%
(is)\,d\rho(s)}{s+ik}\text{ (by (\ref{proj})).}\nonumber
\end{align}
By Theorem \ref{step like refl} $A$ is analytic in $\mathbb{C}^{+}\setminus
i\Delta$. Since $m_{\pm}$ are Herglotz functions, so is $-\left(  m_{-}%
+m_{+}\right)  ^{-1}$ and hence by (\ref{Hergtotz}) $iz\left\{  m_{-}%
(z^{2})+m_{+}(z^{2})\right\}  ^{-1}$ does not grow faster that $z^{3}$ along
the line $\mathbb{R}+ih_{0}$ for any $h_{0}>h$.\footnote{In fact, it is even
bounded.} Due to the rapid decay of $\xi_{x,t}\left(  z\right)  $ along
$\mathbb{R}+ih_{0}$ we can deform the contour in the first integral of the
right hand side of (\ref{cont deform}) to $\mathbb{R}+ih_{0}$. We have%
\begin{align*}
&  \frac{i}{2\pi}\int_{\mathbb{R}}\ \frac{\xi_{x,t}\left(  z\right)  A\left(
z\right)  }{z-\left(  k-i0\right)  }\ dz\text{ \ (by (\ref{drou}))}\\
&  =\frac{i}{2\pi}\int_{\mathbb{R}+ih_{0}}\frac{\xi_{x,t}\left(  z\right)
A\left(  z\right)  }{z-k}\ dz-\int_{0}^{h}\frac{\xi_{x,t}(is)\,d\rho(s)}%
{s+ik}.
\end{align*}
Inserting this formula into (\ref{cont deform}) yields%
\[
\Phi_{x,t}\left(  k\right)  =\frac{i}{2\pi}\int_{\mathbb{R}+ih_{0}}\frac
{\xi_{x,t}\left(  z\right)  A\left(  z\right)  }{z-k}\ dz.
\]
We emphasize that the cancellation of the second integral of the right hand
side of (\ref{cont deform}) is the main reason why our approach works.

Thus we have obtained the crucially important split of our symbol:%
\begin{equation}
\ \widetilde{\varphi}_{x,t}=\varphi_{x,t}^{0}+\Phi_{x,t}, \label{2}%
\end{equation}
where%
\[
\varphi_{x,t}^{\left(  0\right)  }:=\xi_{x,t}R_{0},\ \Phi_{x,t}(k)=\int%
_{\mathbb{R}+ih_{0}}\ \frac{\phi_{x,t}\left(  z\right)  }{z-k}dz,\ \phi
_{x,t}:=\frac{i}{2\pi}\xi_{x,t}A.
\]
Since we can take $h_{0}$ in (\ref{2}) as large as we want, the function
$\Phi_{x,t}$ is entire. Due to the rapid decay of $\xi_{x,t}\left(  z\right)
$ along each $\mathbb{R}+ih_{0}$ one easily sees that $\partial_{t}%
^{n}\partial_{x}^{m}\Phi_{x,t}$ is also entire for any nonnegative integers
$n,m$ and hence by Proposition \ref{trace AAK}%
\begin{equation}
\partial_{t}^{n}\partial_{x}^{m}\mathbb{H}(\Phi_{x,t})=\mathbb{H}(\partial
_{t}^{n}\partial_{x}^{m}\Phi_{x,t})\in\mathfrak{S}_{1}. \label{S_p 1}%
\end{equation}
The Hankel operator $\mathbb{H}(\varphi_{x,t}^{0})$ is implicitly studied in
\cite{CohenKappSIAM87}. In particular, it follows from the proof of Theorem
5.1 in \cite{CohenKappSIAM87} (the main result of this paper) that for $0\leq
m\leq5$
\begin{equation}
\partial_{x}^{m}\mathbb{H}\left(  \varphi_{x,t}^{0}\right)  ,\partial
_{t}\mathbb{H}\left(  \varphi_{x,t}^{0}\right)  \in\mathfrak{S}_{\infty}.
\label{S_infty}%
\end{equation}
Indeed, as was discussed in Subsection \ref{H}, $\mathbb{H}\left(
\varphi_{x,t}^{0}\right)  $ is unitarily equivalent to the integral Hankel
operator given by (\ref{eq4.10}) with the kernel%
\[
H_{x,t}\left(  s\right)  =\frac{1}{\pi}\int_{-\infty}^{\infty}\varphi
_{x,t}^{0}\left(  k\right)  e^{2iks}dk.
\]
Under our conditions on $q_{0}$, by (b) on page 1012 of \cite{CohenKappSIAM87}
we have that $\partial_{x}^{m}H_{x,t}\in L^{1}\left(  0,\infty\right)  ,$
$0\leq m\leq5$. By Corollary 8.11 of \cite{Peller2003} the integral Hankel
operator with kernel $\partial_{x}^{m}H_{x,t}$ is compact\footnote{In fact,
boundedness alone is trivial.} and therefore, by unitary equivalence, so is
$\partial_{x}^{m}\mathbb{H}(\varphi_{x,t}^{0})$. Since $\partial_{t}H_{x,t}=$
$\partial_{x}^{3}H_{x,t}$ one also concludes that $\partial_{t}\mathbb{H}%
(\varphi_{x,t}^{0})\in\mathfrak{S}_{\infty}$. Thus (\ref{S_infty}) is proven
and due to (\ref{S_p 1}) so is (3).

It remains to show (4). Due to (\ref{S_p 1}) one only needs to demonstrate
that $\mathbb{H}(\varphi_{x,t}^{0})$ is trace class. It follows from
(\ref{Rnot}) that%
\[
\varphi_{x,t}^{0}=\xi_{x,t}f_{0}G_{0}.
\]
By Proposition \ref{lemma on Trace class} $\mathbb{H}\left(  \xi_{x,t}%
R_{a}\right)  $ is trace class if each $\mathbb{H}\left(  \xi_{x,t}\right)
,\mathbb{H}\left(  f_{0}\right)  ,\mathbb{H}\left(  G_{0}\right)  $ is. For
$\mathbb{H}\left(  \xi_{x,t}\right)  $ as before we write
\[
\mathbb{H}\left(  \xi_{x,t}\right)  =\mathbb{H}\left(  \Phi_{x,t}^{\left(
0\right)  }\right)  ,\ \ \ \Phi_{x,t}^{\left(  0\right)  }\left(  k\right)
:=\frac{i}{2\pi}\int_{\mathbb{R}+ih_{0}}\frac{\xi_{x,t}\left(  z\right)
}{z-k}\ dz
\]
and hence $\mathbb{H}(\xi_{x,t})$ is trace class for any real $x$ and positive
$t$.

Since $f_{0}\in H_{+}^{\infty}$ (see (\ref{f's})) we simply have
$\mathbb{H}(f_{0})=0$. It remains to show that $\mathbb{H}(G_{0}%
)\in\mathfrak{S}_{1}$. It follows from condition 2 of Hypothesis \ref{Cond}
that $Q\left(  s\right)  =O\left(  s^{-\alpha+1}\right)  ,s\rightarrow\infty,$
and hence
\[
\left\vert \partial_{k}^{2}G_{0}\left(  k\right)  \right\vert \leq4\int%
_{0}^{\infty}s^{2}\left\vert Q\left(  s\right)  \right\vert ds<\infty.
\]
By Proposition \ref{trace AAK} $\mathbb{H}(G_{0})$ is trace class which
completes the proof.
\end{proof}

Theorem \ref{th9.5} says that $I+\mathbb{H}(\varphi_{x,t})$ is invertible
globally in time which is the main reason for validity of the IST for any
step-type data. It's also the most nontrivial part of Theorem \ref{th9.5}. The
proof is based on Theorems \ref{conver of m-functions}, \ref{step like refl},
properties of the m-function discussed in subsection \ref{m-funct}, and subtle
arguments and facts from the theory of Hankel/Toeplitz operators. This has
been incrementally improved in the course of our \cite{RybNON10},
\cite{RybNON11}, \cite{RybCommPDE13}, \cite{GruRyOTAA13}, \cite{GruRyPAMS13},
\cite{GruRyb-prepr-13}.

\begin{remark}
Theorem \ref{th9.5} does not say that the Hankel operator corresponding to
each piece in (\ref{eq9.9}) is trace class. In fact, it is shown in
\cite{GruRyb-prepr-13} that the Hankel operator with symbol $\phi_{x,t}\left(
k\right)  :=\int_{0}^{h}\frac{\xi_{x,t}(is)\,d\rho(s)}{s+ik}$ is trace class
iff $\int_{0}^{h}d\rho(s)/s$ is bounded.
\end{remark}

\subsection{Proof of the main theorem\label{main theorem}}

The proof of Theorem \ref{MainThm} merely combines Theorem \ref{th9.5} and
results of \cite{CohenKappSIAM87}.

\begin{proof}
Take $b<0$ and consider $q_{b}$. Below any object corresponding to $q_{b}$
will be labelled by either a subscript $b$ or a superscript $\left(  b\right)
$. The KdV equation with the initial data $q_{b}$ has a classical solution
$u_{b}\left(  x,t\right)  $ \cite{CohenKappSIAM87} given by the Dyson formula%
\begin{equation}
u_{b}\left(  x,t\right)  =-2\partial_{x}^{2}\log\det\left(  I+\mathbb{H}%
(\varphi_{x,t}^{\left(  b\right)  })\right)  , \label{u0}%
\end{equation}
where by Theorem \ref{th9.5} the operator $\mathbb{H}(\varphi_{x,t}^{\left(
b\right)  })$ is trace class and hence the determinant is well-defined.

Let%
\begin{equation}
u\left(  x,t\right)  =-2\partial_{x}^{2}\log\det\left(  I+\mathbb{H}%
(\varphi_{x,t})\right)  , \label{u}%
\end{equation}
where by the same theorem the operator $\mathbb{H}(\varphi_{x,t})$ is trace
class. Consider
\begin{align}
\Delta u  &  :=u-u_{b}\label{delta u}\\
&  =2\partial_{x}^{2}\log\det\left(  I+\mathbb{H}(\varphi_{x,t})\right)
^{-1}\left(  I+\mathbb{H}(\varphi_{x,t}^{\left(  b\right)  })\right)
\nonumber\\
&  =2\partial_{x}^{2}\log\det\left(  I\mathbb{+}\left(  I+\mathbb{H}%
(\varphi_{x,t})\right)  ^{-1}\left(  \mathbb{H}(\varphi_{x,t}^{\left(
b\right)  })-\mathbb{H}(\varphi_{x,t})\right)  \right) \nonumber\\
&  =2\partial_{x}^{2}\log\det\left\{  I\mathbb{+}\left[  I+\mathbb{H}%
(\varphi_{x,t})\right]  ^{-1}\Delta\mathbb{H}(\varphi_{x,t})\right\}
.\nonumber
\end{align}
By Theorem \ref{th9.5} $\left[  1+\mathbb{H}\left(  \varphi_{x,t}\right)
\right]  ^{-1}$ is a bounded operator independent of $b$. In virtue of
(\ref{fi tilde}) and (\ref{2}) for $\Delta\mathbb{H}(\varphi_{x,t})$ we have%
\begin{align*}
\Delta\mathbb{H(}\varphi_{x,t})  &  \mathbb{=}\mathbb{H}\left(  \varphi
_{x,t}^{\left(  b\right)  }-\varphi_{x,t}\right)  =\mathbb{H}\left(
\widetilde{\varphi}_{x,t}^{\left(  b\right)  }-\widetilde{\varphi}%
_{x,t}\right) \\
&  =\mathbb{H}\left(  \Delta\Phi_{x,t}\right)  ,
\end{align*}
where%
\begin{align*}
\Delta\Phi_{x,t}\left(  k\right)   &  =\frac{i}{2\pi}\int_{\mathbb{R}+ih_{0}%
}\frac{\xi_{x,t}\left(  z\right)  \Delta A\left(  z\right)  }{z-k}\ dz\\
&  =\frac{i}{2\pi}\int_{\mathbb{R}+ih_{0}}2iz\ f_{1}\left(  z\right)
\xi_{x,t}\left(  z\right)  \ \Delta f\left(  z\right)  \frac{dz}{z-k}%
\end{align*}
and
\[
\Delta f\left(  z\right)  :=\left(  m_{-}^{\left(  b\right)  }(z^{2}%
)+m_{+}(z^{2})\right)  ^{-1}-\left(  m_{-}(z^{2})+m_{+}(z^{2})\right)  ^{-1}.
\]
By (\ref{S_p 1}) $\partial_{t}^{n}\partial_{x}^{m}\mathbb{H}\left(  \Delta
\Phi_{x,t}\right)  $ is trace class. We now show that for all $n$ and $m$%
\begin{equation}
\left\Vert \partial_{t}^{n}\partial_{x}^{m}\mathbb{H}\left(  \Delta\Phi
_{x,t}\right)  \right\Vert _{\mathfrak{S}_{1}}\rightarrow0,b\rightarrow
-\infty. \label{limit in trace norm}%
\end{equation}
To this end consider
\begin{align}
\Delta\phi\left(  k\right)   &  :=\partial_{t}^{n}\partial_{x}^{m}\left(
\Delta\Phi_{x,t}\right) \label{delta}\\
&  =\frac{i}{2\pi}\int_{\mathbb{R}+ih_{0}}\partial_{t}^{n}\partial_{x}^{m}%
\xi_{x,t}\left(  z\right)  2iz\ f_{1}\left(  z\right)  \ \Delta f\left(
z\right)  \frac{dz}{z-k}.\nonumber
\end{align}
Differentiating (\ref{delta}) in $k$ twice, one has%
\[
\partial_{k}^{2}\Delta\phi\left(  k\right)  =\frac{1}{\pi}\int_{\mathbb{R}%
+ih_{0}}\left(  2iz\right)  ^{3n+m+1}\xi_{x,t}\left(  z\right)  f_{1}\left(
z\right)  \Delta f\left(  z\right)  \ \frac{dz}{\left(  z-k\right)  ^{3}}%
\]
and hence%
\begin{align}
\sup_{k\in\mathbb{R}}\left\vert \partial_{k}^{2}\Delta\phi\left(  k\right)
\right\vert  &  \leq\frac{2^{3n+m+1}}{\pi h_{0}^{3}}\int_{\mathbb{R}+ih_{0}%
}\left\vert z\right\vert ^{3n+m+1}\left\vert \xi_{x,t}\left(  z\right)
\right\vert \left\vert f_{1}\left(  z\right)  \right\vert \left\vert \Delta
f\left(  z\right)  \right\vert \ \left\vert dz\right\vert \label{est}\\
&  \leq const.\int_{\mathbb{R}+ih_{0}}\left\vert z^{3n+m+1}\xi_{x,t}\left(
z\right)  \right\vert \left\vert \Delta f\left(  z\right)  \right\vert
\ \left\vert dz\right\vert \text{ (since }f_{1}\in H_{+}^{\infty}%
\text{).}\nonumber
\end{align}
Since%
\[
\left\vert \xi_{x,t}(\alpha+ih_{0})\right\vert =\xi_{x,t}(ih_{0})\exp\left\{
-24h_{0}t\alpha^{2}\right\}  ,
\]
one sees that $z^{3n+m+1}\xi_{x,t}\left(  z\right)  $ falls off along
$\mathbb{R}+ih_{0}$ faster than exponential for any $n,m$. Split the contour
$\mathbb{R}+ih_{0}$ into $\gamma_{N}=\left(  -N+ih_{0},N+ih_{0}\right)  $ and
$\Gamma_{N}=\left(  \mathbb{R}+ih_{0}\right)  \setminus\gamma_{N}$. Since
clearly $\operatorname*{Spec}\mathbb{L}_{q_{b}}\geq-h^{2}$ by Theorem
\ref{conver of m-functions} $\Delta f\left(  z\right)  \rightarrow
0,b\rightarrow-\infty,$ uniformly on $\gamma_{N}$ for any $N$ and hence%
\[
\int_{\gamma_{N}}\left\vert z^{3n+m+1}\xi_{x,t}\left(  z\right)  \right\vert
\left\vert \Delta f\left(  z\right)  \right\vert \ \left\vert dz\right\vert
\rightarrow0,\ \ \ b\rightarrow-\infty.
\]
Consider $\Delta f\left(  z\right)  $ on $\Gamma_{N}$. Since the $m$-function
has a non-negative imaginary part\footnote{Note that $z^{2}$ is in
$\mathbb{C}^{+}$ if $z$ is in the first quadrant. If $z$ is in the second
quadrant, then $\operatorname{Im}m\left(  z^{2}\right)  \leq0$.} (the Herglotz
property), one has%
\begin{align*}
\left\vert \Delta f\left(  z\right)  \right\vert  &  \leq\left\vert \frac
{1}{m_{-}^{\left(  b\right)  }(z^{2})+m_{+}(z^{2})}\right\vert +\left\vert
\frac{1}{m_{-}(z^{2})+m_{+}(z^{2})}\right\vert \\
&  \leq\left\vert \frac{1}{\operatorname{Im}m_{-}^{\left(  b\right)  }%
(z^{2})+\operatorname{Im}m_{+}(z^{2})}\right\vert +\left\vert \frac
{1}{\operatorname{Im}m_{-}(z^{2})+\operatorname{Im}m_{+}(z^{2})}\right\vert \\
&  \leq\left\vert \frac{2}{\operatorname{Im}m_{+}(z^{2})}\right\vert .
\end{align*}
For our decay condition at $+\infty$ one has $\psi_{+}\left(  0,z\right)
=1+O\left(  1/z\right)  $ (see e.g. .\cite{DeiftTrub}) and hence $m_{+}\left(
z^{2}\right)  =\partial_{x}\psi_{+}\left(  0,z\right)  /\psi_{+}\left(
0,z\right)  =iz+O\left(  1/z\right)  $ as $\left\vert z\right\vert
\rightarrow\infty$ in $\mathbb{C}^{+}$. Therefore, $\left\vert \dfrac
{2}{\operatorname{Im}m_{+}(z^{2})}\right\vert $ is bounded on $\Gamma_{N}$ and
hence by choosing $N$ large enough the integral%
\begin{align*}
&  \int_{\Gamma_{N}}\left\vert z^{3n+m+1}\xi_{x,t}\left(  z\right)
\right\vert \left\vert \Delta f\left(  z\right)  \right\vert \ \left\vert
dz\right\vert \text{ }\\
&  \leq2\int_{\Gamma_{N}}\left\vert z^{3n+m+1}\xi_{x,t}\left(  z\right)
\right\vert \left\vert \frac{dz}{\operatorname{Im}m_{+}(z^{2})}\right\vert \
\end{align*}
can be made as small as one wishes for any real $x$ and positive $t$.

Thus we can conclude that $\left\Vert \Delta\phi^{\prime\prime}\right\Vert
_{\infty}\rightarrow0$ as $b\rightarrow-\infty$ and Proposition
\ref{trace AAK} implies (\ref{limit in trace norm}).

Next we show that $u\left(  x,t\right)  $ given by (\ref{u}) is differentiable
three time in $x$ and once in $t$. We have
\begin{align*}
u\left(  x,t\right)   &  =-2\partial_{x}^{2}\log\det\left\{  1+\mathbb{H}%
(\varphi_{x,t})\right\}  =-2\partial_{x}^{2}\log\det\left\{  1+\mathbb{H}%
(\widetilde{\varphi}_{x,t})\right\} \\
&  =-2\partial_{x}^{2}\log\det\left\{  1+\mathbb{H}(\varphi_{x,t}^{0}%
+\Phi_{x,t})\right\}  \text{ \ (by (\ref{2}))}\\
&  =-2\partial_{x}^{2}\log\det\left\{  1+\mathbb{H}(\varphi_{x,t}%
^{0})+\mathbb{H}\left(  \Phi_{x,t}\right)  \right\} \\
&  =-2\partial_{x}^{2}\log\det\left\{  1+\mathbb{H}(\varphi_{x,t}^{0})\right\}
\\
&  -2\partial_{x}^{2}\log\det\left\{  1+\left[  1+\mathbb{H}(\varphi_{x,t}%
^{0})\right]  ^{-1}\mathbb{H}\left(  \Phi_{x,t}\right)  \right\} \\
&  =u_{0}\left(  x,t\right)  +U\left(  x,t\right)  \text{ \ (by\ (\ref{u0})
with }b=0\text{),}%
\end{align*}
where%
\[
U\left(  x,t\right)  :=-2\partial_{x}^{2}\log\det\left\{  1+\left[
1+\mathbb{H}(\varphi_{x,t}^{0})\right]  ^{-1}\mathbb{H}\left(  \Phi
_{x,t}\right)  \right\}  .
\]
The well-known differentiation formula
\[
\left(  \log\det\left(  1+A\right)  \right)  ^{\prime}=\operatorname*{tr}%
\left(  1+A\right)  ^{-1}A^{\prime},
\]
(\ref{S_p 1}) and (\ref{S_infty}) imply that $U\left(  x,t\right)  $ is
differentiable three time in $x$ and once in $t$. Since $u_{0}\left(
x,t\right)  $ is the classical solution to (\ref{KdV}) with $q=q_{0}$ by
definition $u_{0}\left(  x,t\right)  $ has the same property and thus so is
$u=u_{0}+U$. It follows from (\ref{delta u}), (\ref{limit in trace norm}), and
Theorem \ref{th9.5} that for $0\leq m\leq3$%
\begin{equation}
\partial_{x}^{m}u_{b}\left(  x,t\right)  \rightarrow\partial_{x}^{m}u\left(
x,t\right)  ,\ \ \ \partial_{t}u_{b}\left(  x,t\right)  \rightarrow\partial
u\left(  x,t\right)  . \label{limits}%
\end{equation}

Finally, it only remains to show that $u$ indeed solves the KdV equation. To
this end, represent $u=u_{b}+\Delta u$ where as above $\Delta u=u-u_{b}$. We
have%
\begin{align}
&  \partial_{t}u-6u\partial_{x}u+\partial_{x}^{3}u\label{rhs}\\
&  =\partial_{t}\Delta u+3\partial_{x}\left[  \left(  \Delta u-2u\right)
\Delta u\right]  +\partial_{x}^{3}\Delta u\nonumber\\
&  \rightarrow0,\ \ b\rightarrow-\infty\text{ (due to (\ref{limits}%
))}\nonumber
\end{align}
and the proof is complete.
\end{proof}

The conditions of Hypothesis \ref{Cond} are very general and admit the case of
$\left\vert R\left(  k\right)  \right\vert =1$ for almost all real $k$ that
has never been considered in the literature before. In the quantum mechanical
sense, such $q$'s are repulsive for plane waves coming from $+\infty$.
Examples include (1) \emph{Gaussian white noise} on a left half line (like the
stock market), (2)\emph{\ Pearson blocks} (certain sparse sequences of bumps),
(3) \emph{Kotani potentials} (certain random slowly decaying at $x\rightarrow
-\infty$ functions \cite{KoU88}), and (4) functions growing at $-\infty$ (not
quite physical), to mention just four.

\begin{remark}
As a by-product, we have shown that the operator-valued function $\left(
x,t\right)  \rightarrow\left[  1+\mathbb{H}(\xi_{x,t}R_{0})\right]
^{-1}\mathbb{H}\left(  \varphi_{x,t}-\xi_{x,t}R_{0}\right)  $ is continuously
differentiable in trace norm five times in $x$ and at least once in $t$.
Analogues statements for $\left(  x,t\right)  \rightarrow\mathbb{H}\left(
\varphi_{x,t}\right)  $ will be studied jointly with S. Grudsky elsewhere. We
only mention here that Proposition \ref{trace AAK} is no longer useful and our
arguments are based upon Peller's subtle characterization of all trace class
Hankel operators \cite{Peller2003} and preliminary results to this effect is
to appear in \cite{RybPAMS17}.
\end{remark}

\begin{remark}
\label{remark}The first condition in Hypothesis \ref{Cond} cannot be relaxed
as the following simple argument suggests. Consider a sequence of soliton type
bumps $q_{n}\left(  x\right)  $ of height $-\varkappa_{n}^{2}$ located on
$\left(  -\infty,0\right)  $ with some phases $\gamma_{n}$. Under the KdV flow
all $q_{n}$ start moving to the right with velocities $2\varkappa_{n}^{2}$. We
can choose $\left(  \varkappa_{n}\right)  ,\left(  \gamma_{n}\right)
\rightarrow\infty$ so that all $q_{n}\left(  x,t\right)  $ would meet at a
fixed point $x_{0}$ at a fixed time $t_{0}$. Apparently, this means that a
\emph{blow-up solution} develops in time $t_{0}$ which can be made arbitrarily
small. The operator $\mathbb{L}_{q}$ is clearly unbounded below. Note that our
approach breaks down in a crucial way if we relax the semiboundedness
condition. We don't plan to pursue this issue any further as this situation
looks physical irrelevant.
\end{remark}

\begin{remark}
The second condition in Hypothesis \ref{Cond} can be somewhat relaxed but the
statement becomes weaker. For example, the condition $\int^{\infty}\left\vert
xq\left(  x\right)  \right\vert dx<\infty$ will guarantee the existence of
$\det\left(  1+\mathbb{H}(\varphi_{x,t})\right)  $ but classical
differentiability becomes a serious issue. Further relaxation of the decay at
$+\infty$ is a big open problem. It follows from the famous 1993 Bourgain
result \cite{Bourgain93} that the problem (\ref{KdV}) will remain well-posed
(although in a much weaker sense) if $q$ is square integrable at $+\infty$ but
it is currently unknown if (\ref{KdV}) is completely integrable for $q\in
L^{2}$. Note that even the particular case of Wigner-von Neumann initial
profiles is still an open problem \cite{MatveevOpenProblems}. But as opposed
to Remark \ref{remark}, such initial profiles are physically relevant as they
may be used to model rogue waves \cite{MatveevOpenProblems}.
\end{remark}

We emphasize that as we have shown the solution (\ref{det_form}) is classical.
That is, the solution is at least three-times continuously differentiable in
$x$ and at least once in $t$ while we don't assume any smooths of the initial
data. Thus the \emph{KdV flow }instantaneously smoothens any (integrable)
singularities of $q\left(  x\right)  $. This effect, commonly called now
\emph{dispersive smoothing}, was first proven in 1978 by Cohen
\cite{Murray(Cohen)78} for box shaped initial data with much of effort. In
\cite{RybCommPDE13} we prove it for any initial data with the decay%
\begin{equation}
q\left(  x\right)  =O\left(  \exp\left\{  -Cx^{\delta}\right\}  \right)
,\ \ \ x\rightarrow\infty, \label{strong decay}%
\end{equation}
with some positive $C$ and $\delta$. If $\delta>1/2$ then $q(x,t)$ is
meromorphic with respect to $x$ on the whole complex plane (with no real
poles) for any $t>0$. If $\delta=1/2$ then $q(x,t)$ is meromorphic in a strip
around the $x-$axis widening proportionally to $\sqrt{t}$. For $0<\delta<1/2$
the solution need not be analytic but is at least Gevrey smooth. We also prove
in \cite{GruRyb-prepr-13} that under the extra condition (\ref{strong decay})
the singular numbers of $\mathbb{H}(\varphi_{x,t})$ have subexponential decay
uniformly on compacts of $\left(  x,t\right)  $. The latter means that the
determinant in (\ref{det_form}) rapidly converges suggesting that
(\ref{det_form}) could be used for numerical computations (cf. recent
\cite{Bornemann10} for new numerical techniques for Fredholm determinants).

We also note that our approach can handle \cite{GruRemRyb2015} nonintegrable
singularities like Dirac $\delta-$functions, Coulomb potentials, etc. and the
strong smoothing effect takes place even in this very singular, although not
quite physical, setting.

Note that our solutions don't in general satisfy conservations laws. It would
be interesting to find an analog of $\int u^{2}\left(  x,t\right)  dx$ under
our conditions. Certain regularizations of conservation laws in a highly
singular setting were considered in our \cite{Ryb10OT}. We also by-passed
answering such questions as well-posedness of the one-sided inverse scattering
problem or direct proof of time evolution of scattering quantities. It is the
limiting procedure that allowed us to detour such delicate questions.

\section{\label{conclusions}Conclusions}

We have given a partial answer to Zakharov's question stated in
\cite{ZakharovetalPhysD2016}: "In spite of all these brilliant achievements,
the theory of the KdV equation is not yet developed to a level which would
satisfy a pragmatic physicist, who may ask the following question: What
happens if the initial data in the KdV equation is neither decaying at
infinity nor periodic? Suppose that the initial data is a bounded function
\[
u(x)=u(x,0),|u(x)|<c.
\]
Can we extend the IST to this case, which has great practical importance?"
Theorem \ref{MainThm} gives the affirmative answer to this question under the
extra assumption that the initial profile decays fast enough at $+\infty$.
However, only boundedness\footnote{In fact, only essential boundedness from
below is required \cite{GruRyb-prepr-13}.} from below is actually required.
This can be viewed as a very strong manifestation of unidirectional nature of
the KdV equation:\ no condition at $-\infty$ and a decay condition at
$+\infty$. A complete answer to Zakharov's question requires the study of the
influence of $+\infty$ on the KdV solution. By Bourgain's theorem a decay
slower than $O\left(  x^{-1/2}\right)  ,x\rightarrow+\infty$, will cause some
major issues as the problem (\ref{KdV}) may fail to be well-posed. But even if
it is well-posed, we need not have even one-sided scattering in this situation
and would have to deal the spectral problem instead. The latter becomes very
complicated and the time evolution of the spectral data need not be simple.
The Lax pair representation of the KdV equation doesn't appear to be any
easier than the KdV equation itself. Due to complexity of the spectrum the
solution may have such a complicated structure that tracking it may be
impossible and example of such situations are already known. It happens in the
study of the so-called \emph{soliton gas}, a random distribution of infinitely
many solitons. The underlying physics of this situation suggests that
statistical description is much more suitable. Such approach was pioneered by
Zakharov \cite{Zakharov(soliton gas)71} back in 1971 and recently received
renewed interest in the connection with \emph{integrable turbulence}
considered in \cite{Zakharov2009}. The theory is under construction. We only
mention \cite{ZakharovetalPhysD2016} where certain type of soliton gas is
described as a closure of the set of reflectionless rapidly decaying
potentials of the Schr\"{o}dinger operator. The resulting solutions are
bounded, but neither periodic nor vanishing as $x\rightarrow\pm\infty$. (see
also Gesztesy et al \cite{Gesztesy-Duke92}). A different approach to soliton
gas and integrable turbulence was put forward by El (see, e.g. \cite{El2016},
\cite{El-et-al2011} and the literature cited therein). His approach is based
on a closure of finite band potentials. Physical examples of integrable
turbulence include coastal areas of seas, and effects occurring in optical fibers.

As we have already mentioned, step like initial profiles were first considered
during the initial boom in the 1970s. The case of $q$'s attaining different
limits at $\pm\infty$ was considered first by Gurevich-Pitaevski
\cite{Gurevich71} in 1973 and has been further developed by Hruslov
\cite{Hruslov76} in 1976, Cohen \cite{Cohen1984} in 1984, Venakides
\cite{Venak86} in 1986, and many others. The most complete asymptotic analysis
of this case was recently done by Teschl and his collaborators in
\cite{TeschlRarefaction16}, \cite{TeschlShock2016} (which also contain the
expensive literature on the subject) \ The treatment is based upon the
scattering theory for step potentials and somewhat similar to the rapidly
decaying case but with serious complications coming from the negative
continuos spectrum. The main feature of this case is that the initial step
will emit solitons which are asymptotically twice as high as the original step
followed by a nearly periodic "washboard". Another physically interesting case
of a profile rapidly decaying at one end and approaching a periodic function
at the other was first considered by Kotlyarov-Hruslov \cite{KK94} in 1994.
The study of such initial profiles recently culminated in \cite{Egorova11}
where two crystals fused together were considered.

Save \cite{Egorova11} our class of step-type initial data is much more
general. However the important problem of finding asymptotics of our solutions
given in Theorem \ref{MainThm} is yet to be solved. The main challenge is that
it is not clear at all how to adapt the powerful machinery of the
Riemann-Hilbert problem so effectively used since the seminal 1993 paper
\cite{DZ93} by Deift-Zhou to our setting. The above mentioned 2016 papers
\cite{TeschlRarefaction16}, \cite{TeschlShock2016} do not suggest an easy solution.

Another important recent breakthrough is related to the 2008 question due to
Deift \cite{DeiftOpenProb08}. He conjectures that, as in the periodic case,
the solution will be almost periodic in time emphasizing that its existence
even for small time is not known. A partial affirmative answer was recently
given by Binder et al \cite{Binderetal2015}.

Note that there are classes of explicit solutions to the KdV equation which
are neither rapidly decaying nor periodic (quasi periodic). Many such
solutions come from considering specific tau-functions in (\ref{tau}). This
way a very important class of positon solutions was discovered by Matveev (see
e.g. \cite{Mat02}). Such solutions are parametrized by a finite number of
constants and have some interesting properties. However they are all singular
and cannot be described within a suitable IST. It has also been long known
(see e.g. the books \cite{AC91} and \cite{MarchBook88}) that certain (formal)
substitutions parametrized by some functions solve the KdV equation but again
neither rapidly decaying nor (quasi) periodic. However as Marchenko says
\cite{MarchenkoWhatIteg91} "It has not been found yet whether it is possible
(and if possible, then by what means) to determine these parameters so as to
obtain the solution satisfying the initial data $u\left(  x,0\right)
=q\left(  x\right)  $, i.e., to solve the Cauchy problem." A partial answer is
given in the same paper \cite{MarchenkoWhatIteg91} in terms of a closure of
certain specific types of potentials. The membership in such classes is hard
to verify. Since time evolution in all these formulas is a priori given and
simple, such solutions are very specific.

\section{Acknowledgement}

We are thankful to Vladimir Zakharov and Gennady El for stimulating
discussions and literature hints. We would also like to thank the anonymous
referees whose constructive criticism and numerous suggestions have
significantly improved the presentation.

\end{document}